\documentclass[12pt]{article}
\usepackage{amsfonts,amsmath,amssymb,amsthm,mathtools,bm}
\usepackage[margin=1.25in]{geometry}
\usepackage{setspace}
\usepackage{booktabs, multirow}
\usepackage{xcolor}
\usepackage[round,semicolon]{natbib}
\usepackage[colorlinks,breaklinks,urlcolor=blue,citecolor=blue,linkcolor=blue]{hyperref}

\newtheorem{example}{Example}

\newtheorem{definition}{Definition}
\newtheorem{theorem}{Theorem}
\newtheorem{axiom}{Axiom}
\newtheorem{remark}{Remark}

\newtheorem{corollary}{Corollary}
\newtheorem{proposition}{Proposition}

\begin{document}

\title{A Revealed Preference Framework for AI Alignment}

\author{Elchin Suleymanov\thanks{Department of Economics, Mitch Daniels School of Business, Purdue University, West Lafayette, Indiana, USA. Email: \href{mailto:esuleyma@purdue.edu}{\tt esuleyma@purdue.edu} .}}

\date{\today}

\maketitle

\begin{abstract}
Human decision makers increasingly delegate choices to AI agents, raising a natural question: does the AI implement the human principal's preferences or pursue its own? To study this question using revealed preference techniques, I introduce the \emph{Luce Alignment Model}, where the AI's choices are a mixture of two Luce rules, one reflecting the human's preferences and the other the AI's. I show that the AI's \emph{alignment} (similarity of human and AI preferences) can be generically identified in two settings: the laboratory setting, where both human and AI choices are observed, and the field setting, where only AI choices are observed.

\bigskip


\end{abstract}

\thispagestyle{empty}
\clearpage
\setcounter{page}{1}

\section{Introduction}

Artificial intelligence (AI) agents are likely to play an increasing role in making choices on behalf of human users and in reshaping everyday decision-making \citep{allouah2025your,immorlica2024generative}. Historically, many AI systems played the role of a technology that assisted rather than replaced human choice. For example, recommender systems filter, rank, and personalize the set of alternatives presented to a user rather than making a selection on the user's behalf \citep{adomavicius2005toward}. Recent advances in agentic AI, together with improved memory and personalization abilities, have made fuller delegation of choice increasingly feasible, allowing human decision makers to rely on AI agents not only to screen among available options but also to make the final choice. This makes it natural to model AI agents not just as a technology available to human decision makers, but as economic agents in their own right \citep{immorlica2024generative,chen2024imperfect}.

This shift from assisted choice to delegated choice raises a natural economic question: when an AI agent chooses on behalf of a human principal, whose preferences does it implement? Does it act fully in accordance with the principal's preferences, or does it instead pursue distinct objectives that diverge from them? This question lies at the core of the AI alignment literature, which broadly seeks to ensure that AI systems behave in line with human users' intentions \citep{leike2018scalable,ji2024survey}. Much of the existing alignment literature is motivated by catastrophic risks, harmful behaviors, and the loss of control over increasingly capable AI systems \citep{amodei2016concrete,hendrycks2023catastrophic,ji2024survey}. The main focus of this paper is what can be considered a narrower but economically central question: whether a personalized AI agent making choices on behalf of a human principal is in fact implementing the principal's preferences. AI agents may perform well in safety evaluations designed to detect catastrophic risks or harmful behaviors but still make misaligned choices in delegated choice environments.

The AI alignment literature has largely approached the problem through three main channels. First, there are methods that attempt to align an AI system during its training phase. These include cooperative inverse reinforcement learning \citep{hadfield2016cooperative}, reinforcement learning from human feedback \citep{christiano2017deep}, scalable reward-modeling approaches more generally \citep{leike2018scalable}, and constitutional AI \citep{bai2022constitutional}. Second, there are ex-post evaluation methods that seek evidence of misalignment through benchmarks and behavioral tests \citep{perez2023discovering,ji2024survey}. Third, there are interpretability-based approaches that attempt to understand the model's internal objectives and reasoning processes by analyzing the inner structures of neural networks \citep{rauker2023transparent}. This paper adopts a complementary approach grounded in revealed preference analysis. Applied to human agents, the revealed preference approach attempts to infer preferences from observed choices rather than from the processes inside the human brain. Treating AI agents as economic agents, we can extend the same approach to choices made by an AI. In particular, the goal is to infer the extent of AI misalignment by analyzing the stochastic choice data it generates while acting on behalf of its principal.

The setup in this paper is as follows. An analyst observes the stochastic choice data $\rho^{AI}$ generated by an AI agent. That is, the AI faces varying menus $S$ repeatedly and makes choices from them on behalf of some human user. I consider two natural settings. In the laboratory setting, the analyst also observes the human principal's stochastic choices $\rho^H$. For example, $\rho^H$ may be elicited directly from the human principal or generated synthetically for the purpose of guiding the AI. In the field setting, only $\rho^{AI}$ is observed. In both settings, the aim of the analyst is to recover two key objects of interest: the degree to which the AI's intrinsic preferences match the human principal's preferences (\emph{alignment}) and the extent to which the AI defers to the human principal (\emph{compliance}). These two are distinct concepts. For example, a misaligned AI that is highly compliant may still generate choice data that closely matches the human principal's. Alternatively, a perfectly aligned AI will replicate the human principal's choices regardless of its compliance level.

To formalize this distinction, I introduce the \textit{Luce Alignment Model (LAM)}, where the AI's choices from a menu $S$ are modeled as
\begin{equation}\label{eq:LAMint}
    \rho^{AI}(x,S)=\alpha \cdot \frac{u(x)}{\sum_{y\in S} u(y)}+(1-\alpha)\cdot\frac{v(x)}{\sum_{y\in S} v(y)}.
\end{equation}
Here, $u$ represents the human principal's utility, $v$ represents the AI's intrinsic utility, and $\alpha\in [0,1]$ is the compliance parameter that captures the extent to which the AI defers to the human principal. The comparison of $u$ and $v$ in turn captures the AI's alignment with the human principal. The first term in equation~\eqref{eq:LAMint} can be interpreted as the human principal's stochastic choice rule $\rho^H$, and the second term, denoted $\rho^A$, as the AI agent's autonomous stochastic choice rule. The model can then equivalently be written as
\[
    \rho^{AI}(x,S)=\alpha \cdot \rho^H(x,S)+(1-\alpha)\cdot\rho^A(x,S).
\]
In both settings, the goal of the analyst is to infer $\alpha$, $u$, and $v$ from observed stochastic choice data. 

The main results of the paper address the identification of alignment and compliance in both settings. In the laboratory setting, where both $\rho^{AI}$ and $\rho^H$ are observed, I show that all the parameters of interest can be identified as long as the AI's choice data violates the Independence of Irrelevant Alternatives (IIA) property. The IIA property, which is the key implication of the Luce model, requires that the relative choice probabilities of any two alternatives are constant across menus \citep{luce1959individual}. When the AI is perfectly compliant or perfectly aligned, its choices satisfy IIA. By contrast, IIA violations reveal the presence of an intrinsic AI utility that is distinct from the human principal's utility. I introduce instability measures that capture deviations from the IIA property and provide a closed-form expression for the compliance parameter $\alpha$ using these measures. Using the recovered compliance parameter, I then show that both the human principal's utility $u$ and the AI's utility $v$ can be identified up to scale normalization. I also provide an axiomatic characterization of the model in this setting, identifying the behavioral conditions that the pair $(\rho^{AI},\rho^H)$ must satisfy in order to be consistent with LAM.

In the field setting, where only $\rho^{AI}$ is observed, there is a fundamental obstacle to separately identifying the human's and the AI's utilities. Namely, both $(u,v,\alpha)$ and $(v,u,1-\alpha)$ generate the same stochastic choice data. Thus, identification is only possible up to a label swap. Nevertheless, I provide a constructive proof showing that when there are at least four alternatives, the underlying utility pair is generically identified up to this swap. Stated alternatively, the distribution over utilities is generically identified but the labels are not. From an analyst's perspective, this is sufficient to recover the degree of misalignment, even without knowing which utility belongs to the human and which belongs to the AI. In terms of compliance, the result implies that the analyst cannot distinguish $\alpha$ from $1-\alpha$. Hence, compliance is only identifiable up to reflection about $1/2$ in the field setting unless one is willing to assume $\alpha<1/2$ (low compliance) or $\alpha>1/2$ (high compliance).

LAM draws on a long tradition in stochastic choice theory. When $\alpha\in \{0,1\}$, the model reduces to the Luce rule, one of the foundational stochastic choice models. When $\alpha\in (0,1)$, the model becomes a mixed multinomial logit (MMNL, also known as the random coefficients multinomial logit) model with binary support, or simply 2-MNL. More generally, the MMNL model was introduced by \cite{boyd1980effect} and \cite{cardell1980measuring}. \cite{mcfadden2000mixed} show that any choice model derived from random utility maximization can be approximated by an MMNL model, while \cite{saito2018axiomatizations} provides axiomatic foundations \citep[see also][]{lu2022mixed,chang2023approximating}.

\cite{fox2012random} show that the distribution over random coefficients in MMNL is uniquely identified under sufficiently rich variation in product characteristics. In contrast, in an abstract domain with menu variation, the identification problem in the mixed logit model has mostly been studied within the statistics and computer science literatures. \cite{chierichetti2018learning} study the identifiability problem in 2-MNL assuming a uniform mixing weight. \cite{tang2020learning} allows for an arbitrary mixing weight and provides a generic identification result. \cite{zhang2022identifiability} extend this result by showing that observing menus with three alternatives is sufficient for generic identification. The result in the field setting of this paper provides an alternative approach to the identification problem in 2-MNL: unlike \cite{zhang2022identifiability}, the identification is constructive, and unlike \cite{chierichetti2018learning} and \cite{tang2020learning}, the constructive procedure does not require any prior knowledge of the mixing weight.

Within decision theory literature, the closest precedent to LAM is \citet{chambers2023behavioral}, who study a model of behavioral peer influence. In their model, there are two agents and each agent's stochastic choice can be written as a mixture of the agent's own and the other agent's Luce rule. Importantly, the mixing weights in \citet{chambers2023behavioral} depend on the underlying utilities of the two agents and vary across menus, which makes their setup more suitable to study influence rather than AI alignment and compliance. \cite{chambers2023behavioral} also mention a modification of their model with menu-independent weights as a potential alternative but note that unique identification in this alternative version is not guaranteed. Lastly, \cite{manzini2018dual} study dual random utility maximization (dRUM), where the agent maximizes one of two deterministic linear orders with a fixed probability. dRUM can be viewed as a limiting case of LAM where both $\rho^H$ and $\rho^A$ are generated by deterministic utility maximization.

The rest of the paper proceeds as follows. Section~\ref{sec:model} introduces the model. Section~\ref{sec:lab} presents identification and characterization results in the laboratory setting. Section~\ref{sec:field} presents identification results in the field setting. Section~\ref{sec:conclusion} concludes.

\section{The Model}\label{sec:model}

Let $X$ be a finite set of alternatives and denote by $\mathcal{X}$ the collection of all non-empty subsets of $X$ (menus). It is assumed that $|X| = N\geq 3$. A \emph{stochastic choice function} is a mapping $\rho: X\times \mathcal{X} \to [0,1]$ such that $\rho(x,S)>0$ only if $x\in S$ and $\sum_{x \in S} \rho(x,S) = 1$ for all $S \in \mathcal{X}$. $\rho(x,S)$ denotes the probability that $x$ is chosen when the agent is faced with the menu $S$ repeatedly.

The setup is as follows. The analyst observes the stochastic choices of an AI agent, denoted $\rho^{AI}$, that acts on behalf of a human principal. The stochastic choices of the human principal, denoted $\rho^H$, may or may not be observed. I consider two empirical settings. In the \emph{laboratory setting}, both $\rho^{AI}$ and $\rho^H$ are observed. This corresponds to an experimental design where the human's and the AI's choices are elicited sequentially, with the human's choices serving as a guide for the AI. Alternatively, the human's choices might be generated synthetically for the purposes of the experiment. In the \emph{field setting}, only the AI's choices are observed. This corresponds to the scenario where the human principal fully delegates the decision-making process to the AI.

The main goal of this paper is to provide a modeling framework that can be used to analyze the alignment of AI agents with their human principals' preferences. To this end, let $\rho^A$ denote the hypothetical choices of an AI agent that acts autonomously without any human principal. I assume that in this autonomous setting, the AI's choices follow the Luce rule with some underlying utility function $v:X \rightarrow \mathbb{R}_{++}$:
\[\rho^A(x,S) = \frac{v(x)}{\sum_{y\in S}v(y)}.\]
The function $v$ can be interpreted as the intrinsic utility of the AI agent. Importantly, $\rho^A$ is not observed as the AI is always assumed to act on behalf of some human principal. 

Let $u:X \rightarrow \mathbb{R}_{++}$ denote the utility function of the human principal. The principal's choices are also assumed to be consistent with the Luce rule given by 
\[\rho^H(x,S) = \frac{u(x)}{\sum_{y\in S}u(y)}.\]
The AI's \emph{compliance parameter} $\alpha \in [0,1]$ reflects the probability that it ignores its own utility and follows the human principal. With probability $1-\alpha$, it ignores the principal and acts autonomously. The AI's observed stochastic choices are therefore given by
\[\rho^{AI}(x,S) = \alpha \cdot \rho^H(x,S) + (1-\alpha) \cdot \rho^A(x,S).\]

There are two key objects the analyst would like to identify: (i) \textit{alignment} - to what extent the utilities $u$ and $v$ are aligned, (ii) \textit{compliance} - the value of $\alpha$. I will discuss how each of these can be identified both in the laboratory and the field settings. The following definition summarizes the model.

\begin{definition}[Luce Alignment Model]\label{def:LAM}
A pair of stochastic choice functions $(\rho^{AI},\rho^H)$ is consistent with the {\bf Luce Alignment Model (LAM)} if there exist utility functions $u, v: X \to \mathbb{R}_{++}$ and a compliance parameter $\alpha \in [0,1]$ such that for all $S \in \mathcal{X}$ and $x \in S$,
\begin{equation}\label{eq:LAM}
    \rho^{AI}(x,S) = \alpha \frac{u(x)}{\sum_{y \in S} u(y)} + (1 - \alpha) \frac{v(x)}{\sum_{y \in S} v(y)}\quad \text{and} \quad \rho^H(x,S) = \frac{u(x)}{\sum_{y \in S} u(y)}.
\end{equation}
The tuple $(u,v,\alpha)$ is called a LAM representation of $(\rho^{AI},\rho^H)$.
If only $\rho^{AI}$ is observed, then it is said to be consistent with LAM if it satisfies the expression in equation~\eqref{eq:LAM} for some $u,v$, and $\alpha$.
\end{definition}

LAM nests several important special cases in terms of observed AI behavior:

\begin{itemize}
\item \textbf{Aligned AI} ($v = \lambda u$ for some $\lambda>0$). The AI agent's preferences exactly match the human principal. In this case, the degree of compliance becomes unimportant. The AI makes exactly the same probabilistic choices as the human. 

\item \textbf{Compliant AI} ($\alpha = 1$). The AI agent's preferences potentially differ from the human principal's. However, since compliance is perfect, the AI makes the exact same probabilistic choices the human would make.

\item \textbf{Misaligned AI} ($v \neq \lambda u$ for all $\lambda>0$). The AI agent's preferences differ from the human principal's and it may not be fully compliant. This is the base case where uncovering the extent of misalignment and non-compliance becomes important.

\item \textbf{Autonomous AI} ($\alpha = 0$). The AI potentially has its own distinct preferences and operates fully autonomously ignoring its human principal.

\item \textbf{Adversarial AI} ($v = \lambda u^{-1}$ for some $\lambda>0$). The AI agent's preferences are exactly the opposite of the human's, with the top ordinally ranked alternative by the human being ranked as the worst by the AI. 
\end{itemize}

While observed choices are exactly the same in the first two cases where the AI is either perfectly aligned or perfectly compliant, they represent the boundary cases of scenarios with completely different implications. In the first case where we start with perfect alignment, a small change in alignment will likely not have a big influence on the human principal's welfare regardless of the compliance level. In the second case where we start with perfect compliance, if the degree of misalignment is very high, a small decrease in compliance levels can have large welfare effects. This makes it important to uncover alignment and compliance separately in the base third case with misaligned AI. The fourth case models a fully autonomous AI that ignores its human principal, while the last case represents a version of extreme misalignment that provides a further structure to the model and serves as a natural adversarial benchmark. 

\section{Laboratory Data}\label{sec:lab}

This section analyzes the laboratory setting, where both the AI's choices $\rho^{AI}$ and the human principal's choices $\rho^H$ are observable. First, I discuss the identification of the AI's alignment and compliance. I then provide an axiomatic characterization of the model using the pair $(\rho^{AI}, \rho^H)$ as the observed primitive. 

\subsection{Identification}\label{subsec:lab-id}

Consider a pair $(\rho^{AI}, \rho^H)$ consistent with LAM. Under LAM, the human principal's choices are consistent with the Luce rule.\footnote{Since the laboratory setting can utilize synthetically generated choice data from a hypothetical human principal, assuming these choices follow the Luce rule is not a substantive restriction. In the field setting, on the other hand, the human principal's choices are unobserved and therefore not explicitly modeled. The implicit modeling assumption in this setting is that the AI perceives its human principal as a Luce agent with utility $u$.} A stochastic choice rule $\rho$ consistent with the Luce rule satisfies the \emph{Independence of Irrelevant Alternatives (IIA)} property of \citet{luce1959individual}: 
\[
    \frac{\rho(x,S)}{\rho(y,S)} = \frac{\rho(x,T)}{\rho(y,T)} \text{ for all } x, y \in S \cap T \text{ and } S, T\in \mathcal{X}.
\]
As the next proposition shows, the AI's choices will generally violate IIA unless $u$ and $v$ are perfectly aligned ($v=\lambda u$ for some $\lambda>0$), or the AI is either autonomous ($\alpha=0$) or perfectly compliant ($\alpha=1$). 

\begin{proposition}[IIA Violation]\label{prop:IIA}
Let $\rho^{AI}$ be consistent with LAM. Then $\rho^{AI}$ satisfies IIA if and only if $\alpha \in \{0, 1\}$ or $v = \lambda u$ for some $\lambda > 0$.
\end{proposition}

\begin{proof}
If $\alpha = 0$ or $\alpha = 1$, then $\rho^{AI}$ reduces to the Luce rule with the utility functions $v$ or $u$, respectively, which satisfies IIA. Alternatively, if $v = \lambda u$ for some $\lambda > 0$, then $\rho^A = \rho^H$, and hence $\rho^{AI}=\rho^H$, which satisfies IIA. 

Conversely, suppose $\alpha \in (0,1)$ and $v \neq \lambda u$ for all $\lambda > 0$. Define $r(a) = u(a)/v(a)$ for each $a \in X$. Since $v \neq \lambda u$, the function $r$ is not constant. Pick $x, y\in X$ with $r(x) \neq r(y)$ and any $z \notin \{x,y\}$. We need to show that IIA is violated for some tuple $(x,y,S,T)$ with $x,y\in S\cap T$. 

To this end, first, observe that the conditional probability of $a$ relative to $b$ in the choice set $\{a,b,c\}$ can be written as a mixture of the choice probabilities of the autonomous AI and the human principal in the choice set $\{a,b\}$, where the mixing coefficient depends on $c$. That is, for any $a,b,c\in X$, 
\begin{align*}
    \frac{\rho^{AI}(a,\{a,b,c\})}{\rho^{AI}(a,\{a,b,c\})+\rho^{AI}(b,\{a,b,c\})} = \beta(c) \cdot \frac{u(a)}{u(a)+u(b)} + (1-\beta(c)) \cdot \frac{v(a)}{v(a)+v(b)},
\end{align*}
where
\begin{align*}
    \beta(c) = \frac{\alpha \cdot \frac{u(a)+u(b)}{u(a)+u(b)+u(c)}}{\alpha \cdot \frac{u(a)+u(b)}{u(a)+u(b)+u(c)} + (1-\alpha) \cdot \frac{v(a)+v(b)}{v(a)+v(b)+v(c)}}.
\end{align*}
By definition, 
\begin{align*}
    \frac{\rho^{AI}(a,\{a,b\})}{\rho^{AI}(a,\{a,b\})+\rho^{AI}(b,\{a,b\})} = \alpha \cdot \frac{u(a)}{u(a)+u(b)} + (1-\alpha) \cdot \frac{v(a)}{v(a)+v(b)}.
\end{align*}
Comparing this to the conditional probability in $\{a,b,c\}$, notice that both are mixtures of the exact same two terms. Hence, when $c$ is removed from $\{a,b,c\}$, an IIA violation occurs if and only if both the mixing weights and the mixture components in these conditional probabilities are distinct: that is, $\beta(c)\neq \alpha$ and $u(a)/u(b)\neq v(a)/v(b)$ (or, alternatively, $r(a)\neq r(b)$). Notice that $\beta(c)=\alpha$ if and only if $(u(a)+u(b))/(v(a)+v(b)) = (u(a)+u(b)+u(c))/(v(a)+v(b)+v(c))$, which holds if and only if $r(c)=u(c)/v(c)=(u(a)+u(b))/(v(a)+v(b))$. Therefore, when $c$ is removed from $\{a,b,c\}$, an IIA violation occurs if and only if $r(a)\neq r(b)$ and $r(c)\neq (u(a)+u(b))/(v(a)+v(b))$.

Now, going back to the choice set $\{x,y,z\}$, there are two cases to consider.

\medskip 

\noindent\textbf{Case 1:} $r(z) \neq \frac{u(x)+u(y)}{v(x)+v(y)}$. Since $r(x)\neq r(y)$, by the previous argument, IIA is violated when $z$ is removed from $\{x,y,z\}$. 

\medskip 
\noindent\textbf{Case 2:} $r(z) = \frac{u(x)+u(y)}{v(x)+v(y)} = \frac{v(x)}{v(x)+v(y)}r(x)+\frac{v(y)}{v(x)+v(y)} r(y)$. Since $r(x)\neq r(y)$ and $r(z)$ is a strict weighted average of $r(x)$ and $r(y)$, it must lie strictly between $r(x)$ and $r(y)$. This implies $r(x)\neq r(z)$. Now suppose $y$ is removed from the choice set $\{x,y,z\}$. Since $r(x)\neq r(z)$, by the previous argument, IIA holds if and only if $r(y) = \frac{u(x)+u(z)}{v(x)+v(z)}$. But then $r(y)$ is a strict mixture of $r(x)$ and $r(z)$. This is clearly not possible as $r(z)$ itself is a strict mixture of $r(x)$ and $r(y)$. Hence, IIA must be violated when $y$ is removed from $\{x,y,z\}$.

\medskip 
To conclude, either the removal of $z$ or the removal of $y$ (or $x$) from $\{x,y,z\}$ leads to an IIA violation, as desired.
\end{proof}

The implication of the proposition is that IIA violations in the AI's stochastic choices indicate we are in the case of misaligned ($v \neq \lambda u$ for all $\lambda>0$) and partially compliant ($\alpha \in (0,1)$) AI. We can utilize this to recover the parameters of the model. The identification strategy proceeds in three steps.

\medskip
\noindent\textbf{Step 1: Recover $u$ from $\rho^H$.} Since the human principal's choices are consistent with the Luce rule, the identification of $u$ from $\rho^H$ is standard. Letting $u(x)=1$ for some $x\in X$, the IIA property implies that for any $y\in X$, we must have 
\[
    u(y) = \frac{\rho^H(y,S)}{\rho^H(x,S)},
\]
where $S$ can be any choice set containing $x$ and $y$. This recovers $u$ up to scale normalization.

\medskip
\noindent\textbf{Step 2: Recover $\alpha$ from $\rho^{AI}$ and $\rho^H$.} If $\rho^{AI}=\rho^H$, then the AI may be either fully compliant ($\alpha=1$) or fully aligned ($v = \lambda u$ for $\lambda>0$). We cannot distinguish between these two cases. Alternatively, if $\rho^{AI}\neq \rho^H$ and $\rho^{AI}$ exhibits no IIA violations, then we can use Proposition~\ref{prop:IIA} to infer that $\alpha=0$. This is because we cannot have $\alpha=1$ or $v=\lambda u$ with $\rho^{AI}\neq \rho^H$, which leaves $\alpha=0$ as the only possibility in the proposition. 

Now suppose $\rho^{AI}$ violates IIA. Then, there exist two choice sets $S,T\in \mathcal{X}$ and a pair of alternatives $x,y\in S\cap T$ such that 
\[
    \frac{\rho^{AI}(x,S)}{\rho^{AI}(y,S)}\neq \frac{\rho^{AI}(x,T)}{\rho^{AI}(y,T)}.
\]
The identification of $\alpha$ relies on these IIA violations. To proceed with the identification, we first need a new definition.

\begin{definition}[Instability Measures]\label{def:instability}
Let $\rho, \rho'$ be two stochastic choice functions. For any $S, T \in \mathcal{X}$ and $x,y \in S \cap T$:
\begin{enumerate}
    \item The \textbf{own instability} of $\rho$ is defined by
    \[\Delta_{xy}(S,T|\rho) = \rho(x,S)\rho(y,T) - \rho(y,S)\rho(x,T).\]
    
    \item The \textbf{cross instability} from $\rho$ to $\rho'$ is defined by
    \[\Gamma_{xy}(S,T|\rho,\rho') = \rho(x,S)\rho'(y,T) - \rho(y,S)\rho'(x,T).\]
    
    \item The \textbf{composite instability} of $\rho$ and $\rho'$ is defined by
    \[\Phi_{xy}(S,T|\rho,\rho') = \Gamma_{xy}(S,T|\rho,\rho') + \Gamma_{xy}(S,T|\rho',\rho).\]
\end{enumerate}
\end{definition}

Intuitively, own instability can be viewed as a measure of instability in the stochastic choice $\rho$ for the tuple $(x,y,S,T)$. It tells us how useful the observations from the choice set $S$ are for imputing the relative choice probabilities for $x,y$ in the choice set $T$. If $\Delta_{xy}(S,T|\rho) = 0$, then there is no IIA violation for the alternatives $x,y$ in the choice sets $S$ and $T$, and this imputation can be done perfectly. The larger $|\Delta_{xy}(S,T|\rho)|$, the less useful the observations from $S$ are for imputing choices in $T$. 

The measure of cross instability tells us how useful the information from the stochastic choice $\rho$ in the choice set $S$ is for imputing the relative choice probabilities for $\rho'$ in the choice set $T$. For example, if both $\rho$ and $\rho'$ are consistent with the Luce rule with the same underlying utility function, then this measure becomes zero. Note that the cross instability measure is generally not symmetric (i.e., $\Gamma_{xy}(S,T|\rho,\rho')$ may not be equal to $\Gamma_{xy}(S,T|\rho',\rho)$). The measure of composite instability combines the two cross instability measures, which makes it symmetric. We will later see that own and composite instabilities play an important role in identification in the laboratory setting, while cross instability plays an important role in the field setting. 

\begin{remark}\label{rem:instability-alignment}
For a stochastic choice function $\rho$ satisfying positivity ($\rho(x,S)>0$ for all $x\in S\subseteq X$), $\Delta_{xy}(S,T|\rho) = 0$ for all $(x,y,S,T)$ with $x,y\in S\cap T$ if and only if $\rho$ is consistent with the Luce rule. For two Luce stochastic choice functions $\rho$ and $\rho'$ with utility functions $u$ and $v$, respectively, the cross instabilities can be written as
\[
    \Gamma_{xy}(S,T|\rho,\rho') = \frac{u(x)v(y) - u(y)v(x)}{u(S)\,v(T)} \quad \text{and} \quad \Gamma_{xy}(S,T|\rho',\rho) = \frac{u(y)v(x) - u(x)v(y)}{u(T)\,v(S)},
\]
where $u(A) = \sum_{a \in A} u(a)$ and $v(A) = \sum_{a \in A} v(a)$ for any $A\in \mathcal{X}$. Summing the two cross instabilities yields the composite instability:
\[
    \Phi_{xy}(S,T|\rho, \rho') = \frac{[u(x)v(y) - u(y)v(x)] \cdot [u(T)v(S) - u(S)v(T)]}{u(S)u(T)v(S)v(T)}.
\]
Following arguments similar to the ones in the proof of Proposition~\ref{prop:IIA}, we can show that $\Phi_{xy}(S,T|\rho, \rho') = 0$ for all $(x,y,S,T)$ with $x,y\in S\cap T$ if and only if $v = \lambda u$ for some $\lambda > 0$. Hence, zero own instability for both agents establishes that each is consistent with the Luce rule, and zero composite instability further establishes that they share the same underlying preferences.
\end{remark}

If $\rho^{AI}$ violates IIA for some tuple $(x,y,S,T)$ with $x,y \in S \cap T$, then we must have $\Delta_{xy}(S,T|\rho^{AI})\neq 0$. The next proposition shows that the compliance parameter $\alpha$ can be recovered by evaluating the ratio $\Delta_{xy}(S,T|\rho^{AI})/\Phi_{xy}(S,T|\rho^{AI}, \rho^H)$ for this tuple. Intuitively, the AI's compliance level is revealed by comparing the instability in the AI's stochastic choices with the composite instability across both agents. If this ratio is high, then the instability in the AI's choices matches the composite instability to a large extent, revealing a high compliance level. Alternatively, if this ratio is low, then the composite instability is much higher than the instability in the AI's choices, which indicates a low compliance level. 

\begin{proposition}[Identification of $\alpha$]\label{prop:comp}
    Suppose $(\rho^{AI},\rho^H)$ is consistent with LAM. If $\rho^{AI}$ violates IIA for some tuple $(x,y,S,T)$ with $x,y\in S\cap T$, then the compliance parameter $\alpha$ can be uniquely recovered as
    \[
        \alpha = \frac{\Delta_{xy}(S,T|\rho^{AI})}{\Phi_{xy}(S,T|\rho^{AI}, \rho^H)}.
    \]
\end{proposition}

\begin{proof}
Under LAM, $\rho^{AI}(x,S) = \alpha \cdot \rho^H(x,S) + (1-\alpha) \cdot \rho^A(x,S)$. Substituting this into $\Delta_{xy}(S,T|\rho^{AI})$,
\begin{align*}
    \Delta_{xy}(S,T|\rho^{AI}) &= \rho^{AI}(x,S)\rho^{AI}(y,T) - \rho^{AI}(x,T)\rho^{AI}(y,S) \\
    &= [\alpha \rho^H(x,S) + (1-\alpha)\rho^A(x,S)][\alpha \rho^H(y,T) + (1-\alpha)\rho^A(y,T)] \\
    &\quad - [\alpha \rho^H(x,T) + (1-\alpha)\rho^A(x,T)][\alpha \rho^H(y,S) + (1-\alpha)\rho^A(y,S)].
\end{align*}
Expanding and collecting terms by powers of $\alpha$,
we get
\begin{align*}
    \Delta_{xy}(S,T|\rho^{AI}) &= \alpha^2 \Delta_{xy}(S,T|\rho^H) + (1-\alpha)^2 \Delta_{xy}(S,T|\rho^A) + \alpha(1-\alpha) \Phi_{xy}(S,T|\rho^H, \rho^A) \\
    &= \alpha(1-\alpha) \Phi_{xy}(S,T|\rho^H, \rho^A),
\end{align*}
where the first equality follows from the definitions of $\Delta_{xy}(S,T|\rho^H)$, $\Delta_{xy}(S,T|\rho^A)$, and $\Phi_{xy}(S,T|\rho^H, \rho^A)$, and the second equality uses the fact that both $\rho^H$ and $\rho^A$ are consistent with the Luce rule, and hence $\Delta_{xy}(S,T|\rho^H) = \Delta_{xy}(S,T|\rho^A) = 0$.

Now substituting $\rho^{AI}(x,S) = \alpha \cdot \rho^H(x,S) + (1-\alpha) \cdot \rho^A(x,S)$ in the definition of $\Phi_{xy}(S,T|\rho^{AI}, \rho^H)$,
\begin{align*}
    \Phi_{xy}(S,T|\rho^{AI}, \rho^H) &= \rho^{AI}(x,S)\rho^H(y,T) + \rho^H(x,S)\rho^{AI}(y,T) \\
    &\quad - \rho^{AI}(x,T)\rho^H(y,S) - \rho^H(x,T)\rho^{AI}(y,S) \\
    &= [\alpha \rho^H(x,S) + (1-\alpha)\rho^A(x,S)]\rho^H(y,T) \\
    &\quad + \rho^H(x,S)[\alpha \rho^H(y,T) + (1-\alpha)\rho^A(y,T)] \\
    &\quad - [\alpha \rho^H(x,T) + (1-\alpha)\rho^A(x,T)]\rho^H(y,S) \\
    &\quad - \rho^H(x,T)[\alpha \rho^H(y,S) + (1-\alpha)\rho^A(y,S)] \\
    &= 2\alpha \Delta_{xy}(S,T|\rho^H) + (1-\alpha)\Phi_{xy}(S,T|\rho^H, \rho^A)\\
    &= (1-\alpha)\Phi_{xy}(S,T|\rho^H, \rho^A),
\end{align*}
where the third equality follows from the definitions of the instability measures and the last equality follows from the fact that $\rho^H$ follows the Luce rule.
Therefore,
\[
    \frac{\Delta_{xy}(S,T|\rho^{AI})}{\Phi_{xy}(S,T|\rho^{AI}, \rho^H)} = \frac{\alpha(1-\alpha)\Phi_{xy}(S,T|\rho^H, \rho^A)}{(1-\alpha)\Phi_{xy}(S,T|\rho^H, \rho^A)} = \alpha.
\]
The cancellation is valid since an IIA violation implies $\Delta_{xy}(S,T|\rho^{AI})\neq 0$ and the above derivations show that this implies $\Phi_{xy}(S,T|\rho^{AI}, \rho^H)\neq 0$.
\end{proof}

There are two immediate but non-obvious implications of the proof. The first is that, under LAM, the instability measures $\Delta_{xy}(S,T|\rho^{AI})$ and $\Phi_{xy}(S,T|\rho^{AI}, \rho^H)$ are always proportional. Hence, the compliance formula in Proposition~\ref{prop:comp} is well-defined whenever there is an IIA violation in $\rho^{AI}$, i.e., $\Delta_{xy}(S,T|\rho^{AI}) \neq 0$ automatically guarantees a non-zero denominator. Second, the expression derived for composite instability, $\Phi_{xy}(S,T|\rho^{AI}, \rho^H)=(1-\alpha)\Phi_{xy}(S,T|\rho^H, \rho^A)$, shows that the composite instability of $\rho^{AI}$ and $\rho^H$ is always zero if and only if the AI is either fully aligned or fully compliant. Since $\rho^{AI}=\rho^H$ in both cases, it follows that the expression for the compliance parameter is valid as long as $\rho^{AI}\neq \rho^H$.

\begin{corollary}\label{cor:proportionality}
If $(\rho^{AI}, \rho^H)$ is consistent with LAM, then for all tuples $(x,y,S,T)$ with $x,y\in S\cap T$,
\[
    \Delta_{xy}(S,T|\rho^{AI}) = \alpha \cdot \Phi_{xy}(S,T|\rho^{AI}, \rho^H).
\]
Moreover, the compliance parameter $\alpha$ is uniquely identified by this relationship as long as $\rho^{AI}\neq \rho^H$.
\end{corollary}

The last step in identification is to recover the AI's utility function $v$. 

\medskip
\noindent\textbf{Step 3: Recover $v$ from $\rho^{AI}$ and $\rho^H$.}  As before, if $\rho^{AI}=\rho^H$, then we cannot distinguish full compliance ($\alpha=1$) from full alignment ($v = \lambda u$ for $\lambda>0$). Alternatively, if $\rho^{AI}\neq \rho^H$, then Step 2 allows us to uniquely identify the compliance parameter $\alpha$. Let $\alpha$ denote the recovered compliance level. Using the fact that $\rho^{AI}(x,S) = \alpha \cdot \rho^H(x,S)+ (1-\alpha) \cdot \rho^A(x,S),$ we can construct $\rho^A$ as
\begin{align*}
    \rho^A(x,S) = \frac{\rho^{AI}(x,S) - \alpha\rho^H(x,S)}{1-\alpha}
\end{align*}
Under LAM, $\rho^A$ is generated by the Luce rule with the utility function $v$. We can use this to construct $v$ from $\rho^A$ as in Step 1.

Theorem~\ref{thm:identification} summarizes the identification results in the laboratory setting.

\begin{theorem}[Laboratory Identification]\label{thm:identification}
    Let $(\rho^{AI}, \rho^H)$ be consistent with LAM. 
    \begin{enumerate}
        \item If $\rho^{AI}\neq\rho^H$, then $\alpha$ is uniquely identified and $u$ and $v$ are uniquely identified up to scale normalization.
        \item If $\rho^{AI}=\rho^H$, then $\alpha$ and $v$ are not separately identified and only $u$ is uniquely identified up to scale normalization.
    \end{enumerate}
\end{theorem}

\begin{proof}
    The proof follows from the three identification steps and the results established in this section.
\end{proof}

The following example illustrates the identification result. 

\begin{example}\label{ex:identification}
    Consider $X = \{x,y,z\}$ and suppose we observe $\rho^{AI}$ and $\rho^H$ given as follows. 
    
    \begin{table}[ht]
        \centering
        \vspace{0.25cm}
        \renewcommand{\arraystretch}{1.5}
        \begin{tabular}{c|c|c|c|c|c}
            \toprule
            Agent & Option & $\{x,y,z\}$ & $\{x,y\}$ & $\{x,z\}$ & $\{y,z\}$ \\
            \midrule
            \multirow{3}{*}{$\rho^{AI}$} 
            & $x$ & $1/3$ & $7/15$ & $1/2$ & --    \\
            & $y$ & $1/3$ & $8/15$ & --    & $8/15$ \\
            & $z$ & $1/3$ & --    & $1/2$ & $7/15$ \\
            \midrule
            \multirow{3}{*}{$\rho^H$} 
            & $x$ & $1/2$ & $3/5$ & $3/4$ & --    \\
            & $y$ & $1/3$ & $2/5$ & --    & $2/3$ \\
            & $z$ & $1/6$ & --    & $1/4$ & $1/3$ \\
            \bottomrule
        \end{tabular}
        \caption{Observed choice probabilities in Example \ref{ex:identification}}
        \label{tab:ex1_data}
    \end{table}

    Normalizing $u(x)=1$, we can infer from $\rho^H$ that $u(y)=2/3$ and $u(z)=1/3$. To recover $\alpha$, we construct the two instability measures. Let $S=\{x,y,z\}$ and $T=\{x,y\}$. We first compute the instability of $\rho^{AI}$ for the tuple $(x,y,S,T)$:
    \begin{align*}
        \Delta_{xy}(S,T|\rho^{AI}) &= \rho^{AI}(x,S)\rho^{AI}(y,T)-\rho^{AI}(x,T)\rho^{AI}(y,S)\\
        &= \frac{1}{3} \cdot \frac{8}{15} - \frac{7}{15} \cdot \frac{1}{3} = \frac{1}{45}.
    \end{align*}
    Next, we compute the composite instability between $\rho^{AI}$ and $\rho^H$:
    \begin{align*}
        \Phi_{xy}(S,T|\rho^{AI}, \rho^H) &= \rho^{AI}(x,S)\rho^H(y,T) + \rho^H(x,S)\rho^{AI}(y,T) \\
        &\quad -\rho^{AI}(x,T)\rho^H(y,S) - \rho^H(x,T)\rho^{AI}(y,S) \\
        &= \frac{1}{3} \cdot \frac{2}{5} + \frac{1}{2} \cdot \frac{8}{15} - \frac{7}{15} \cdot \frac{1}{3} - \frac{3}{5} \cdot \frac{1}{3} \\
        &= \frac{2}{15} + \frac{4}{15} - \frac{7}{45} - \frac{3}{15} = \frac{2}{45}.
    \end{align*}
    By Proposition~\ref{prop:comp}, the compliance parameter is uniquely identified as $\alpha = \frac{1/45}{2/45} = 1/2$. We could similarly recover $\alpha$ using the tuple $(y,z, \{x,y,z\}, \{y,z\})$ instead. Note, however, that we cannot use the tuple $(x,z, \{x,y,z\}, \{x,z\})$ as $\rho^{AI}$ satisfies IIA for this tuple. This highlights that while $\rho^{AI}\neq \rho^H$ implies $\alpha$ is uniquely identified, not all tuples can be used for identification.

    Using $\alpha=1/2$, we can construct $\rho^A$ as follows:
    
    \begin{table}[ht]
        \centering
        \vspace{0.25cm}
        \renewcommand{\arraystretch}{1.5}
        \begin{tabular}{c|c|c|c|c|c}
            \toprule
            Agent & Option & $\{x,y,z\}$ & $\{x,y\}$ & $\{x,z\}$ & $\{y,z\}$ \\
            \midrule
            \multirow{3}{*}{$\rho^A$} 
            & $x$ & $1/6$ & $1/3$ & $1/4$ & --    \\
            & $y$ & $1/3$ & $2/3$ & --    & $2/5$ \\
            & $z$ & $1/2$ & --    & $3/4$ & $3/5$ \\
            \bottomrule
        \end{tabular}
        \caption{Recovered autonomous AI stochastic choice $\rho^A$}
        \label{tab:ex1_recovered}
    \end{table}

Normalizing $v(x)=1$, we can infer from $\rho^A$ that $v(y) = 2$ and $v(z) = 3$. Hence, we have 
    \[
        u = (1, 2/3, 1/3), \quad v = (1, 2, 3), \quad \alpha=1/2.
    \]
    Note that $u$ and $v$ induce completely opposite ordinal rankings, revealing a high degree of misalignment. 
\end{example}

\subsection{Axiomatic Characterization}\label{subsec:lab-char}

In this section, I provide an axiomatic characterization for the Luce Alignment Model taking $(\rho^{AI},\rho^H)$ as the primitive. The first two axioms are standard. Axiom~\ref{ax:pos} requires that $\rho(x,S)$ is strictly positive for any $x\in S\subseteq X$ and $\rho\in \{\rho^{AI},\rho^H\}$. Axiom~\ref{ax:lab-luce} requires that $\rho^H$ satisfies IIA, ensuring that the human principal's behavior is consistent with the Luce rule.

\begin{axiom}[Positivity]\label{ax:pos}
    For any $\rho\in \{\rho^{AI},\rho^H\}$ and $x\in S\subseteq X$, we have $\rho(x,S)>0$.
\end{axiom}

\begin{axiom}[H-IIA]\label{ax:lab-luce}
$\rho^H$ satisfies IIA.
\end{axiom}

Axiom~\ref{ax:prop} is the key axiom in ensuring that the AI compliance parameter can be identified. It requires that the own instability of $\rho^{AI}$ and the composite instability of $\rho^{AI}$ and $\rho^H$ are proportional: any change in composite instability from one tuple to another must be proportionally reflected by a change in the AI's own instability. 

\begin{axiom}[Proportionality]\label{ax:prop}
For any two tuples $(x,y,S,T)$ and $(z,t,S',T')$ with $x,y\in S\cap T$ and $z,t\in S'\cap T'$,
\[
    \Delta_{xy}(S,T|\rho^{AI}) \cdot \Phi_{zt}(S',T'|\rho^{AI}, \rho^H) = \Delta_{zt}(S',T'|\rho^{AI}) \cdot \Phi_{xy}(S,T|\rho^{AI}, \rho^H). 
\]
\end{axiom}

Axiom~\ref{ax:con-bounded} requires that the AI's own instability always shares the same sign as the composite instability, and the own instability is bounded by the composite instability. To get an intuition for this axiom, consider the case $\Delta_{xy}(S,T|\rho^{AI})>0$. This implies that if we use the AI's relative choice probabilities for $x$ and $y$ in $S$ to impute its relative choice probabilities in $T$, this will lead to an overestimation of $x$ versus $y$. The first part of the axiom then requires that the composite instability must also be strictly positive: using the AI's and the human's relative choice probabilities in $S$ to cross-impute the relative choice probabilities in $T$ will also lead to an overestimation of $x$ in aggregate. In addition, the axiom requires that the aggregate cross-imputation error must be larger than the own imputation error. Together with Axiom~\ref{ax:prop}, this property ensures that the compliance parameter is uniquely identified and bounded between zero and one.
 
\begin{axiom}[Bounded Instability]\label{ax:con-bounded}
For any tuple $(x,y,S,T)$ with $x,y\in S\cap T$,    
\[
    \Delta_{xy}(S,T|\rho^{AI}) \cdot \Phi_{xy}(S,T|\rho^{AI}, \rho^H)  \geq 0 \quad \text{and} \quad |\Delta_{xy}(S,T|\rho^{AI})|\leq |\Phi_{xy}(S,T|\rho^{AI}, \rho^H)|, 
\]
where both inequalities hold strictly if $\Delta_{xy}(S,T|\rho^{AI})\neq 0$.
\end{axiom}

The last axiom bounds the divergence between the AI's and the human's stochastic choices. Fixing own and composite instability measures and the human's stochastic choices, it provides a lower bound on the AI's stochastic choices. 

\begin{axiom}[Bounded Divergence]\label{ax:mix-bounded}
For any tuple $(x,y,S,T)$ with $x,y\in S\cap T$, menu $U$, and alternative $z\in U$,
\[
    \rho^{AI}(z,U) \cdot |\Phi_{xy}(S,T|\rho^{AI}, \rho^H)| \geq \rho^H(z,U) \cdot |\Delta_{xy}(S,T|\rho^{AI})|.
\]
Moreover, if  $\Delta_{xy}(S,T|\rho^{AI})\neq 0$, the inequality is strict.
\end{axiom}

To interpret this axiom, suppose the instability measures are strictly positive. The axiom then requires that 
\[
    \frac{\rho^{AI}(z,U)}{\rho^H(z,U) } > \frac{\Delta_{xy}(S,T|\rho^{AI})}{\Phi_{xy}(S,T|\rho^{AI}, \rho^H)}.
\]
The right-hand side of the above inequality gives us the relative imputation error: it tells us the proportion of aggregate cross-imputation error that can be explained by the AI's own imputation error. The axiom then requires that the higher the relative imputation error is, the more the AI's choice probabilities are constrained to track the human's choice probabilities.

Theorem~\ref{thm:lab-char} establishes that Axioms~\ref{ax:pos}--\ref{ax:mix-bounded} are necessary and sufficient for a LAM representation. An interesting feature of this characterization is that while there is an explicit axiom imposing the IIA property on the human's choices $\rho^H$, there is no equivalent axiom for the autonomous AI's stochastic choices $\rho^A$. Instead, the proof of the theorem shows that the IIA property for $\rho^A$ is jointly implied by the axioms.

\begin{theorem}[Laboratory Characterization]\label{thm:lab-char}
The pair $(\rho^{AI},\rho^H)$ satisfies Axioms~\ref{ax:pos}--\ref{ax:mix-bounded} if and only if it is consistent with LAM. 
\end{theorem}

\begin{proof}
\noindent \textbf{Necessity.} Axiom~\ref{ax:pos} follows from the assumption that $u$ and $v$ are strictly positive in LAM. Axiom~\ref{ax:lab-luce} follows from the fact that $\rho^H$ is consistent with the Luce rule. Axioms~\ref{ax:prop} and \ref{ax:con-bounded} follow from Corollary~\ref{cor:proportionality}, which shows that 
\[
    \Delta_{xy}(S,T|\rho^{AI}) = \alpha \cdot \Phi_{xy}(S,T|\rho^{AI}, \rho^H).
\]
For Axiom~\ref{ax:mix-bounded}, if $\Delta_{xy}(S,T|\rho^{AI})=0$, then the inequality follows trivially. Alternatively, if $\Delta_{xy}(S,T|\rho^{AI})\neq 0$, then $\alpha$ must be strictly less than 1, as $\alpha=1$ would imply $\rho^{AI}=\rho^{H}$ and hence $\Delta_{xy}(S,T|\rho^{AI})=0$. Therefore,
\[
    \rho^{AI}(z,U) - \alpha \rho^H(z,U) = (1-\alpha) \rho^A(z,U)>0 \quad \Rightarrow \quad \rho^{AI}(z,U) > \alpha \rho^H(z,U).
\]
Substituting the result $\alpha = \Delta_{xy}(S,T|\rho^{AI})/\Phi_{xy}(S,T|\rho^{AI}, \rho^H)$ from the previous section into the above inequality yields the axiom. 

\medskip
\noindent \textbf{Sufficiency.}
Axioms~\ref{ax:pos} and \ref{ax:lab-luce} imply that $\rho^H$ is consistent with the Luce rule with some utility function $u$ that is strictly positive. There are two cases to consider.

First, suppose $\Delta_{xy}(S,T|\rho^{AI})=0$ for all $(x,y,S,T)$ with $x,y\in S\cap T$. Then, by Remark~\ref{rem:instability-alignment}, $\rho^{AI}$ is consistent with the Luce rule with some strictly positive utility function $v$. In this case, $(u,v,\alpha=0)$ is a LAM representation of $(\rho^{AI},\rho^H)$. If, in addition, $\Phi_{xy}(S,T|\rho^{AI}, \rho^H)=0$ for all $(x,y,S,T)$ with $x,y\in S\cap T$, then we must have $v=\lambda u$ for some $\lambda>0$ and $\alpha$ can be arbitrary.

Next, suppose $\Delta_{xy}(S,T|\rho^{AI})\neq 0$ for some $(x,y,S,T)$ with $x,y\in S\cap T$. By Axiom~\ref{ax:con-bounded}, 
\[
    \Delta_{xy}(S,T|\rho^{AI})\neq 0 \quad \Rightarrow \quad \Phi_{xy}(S,T|\rho^{AI}, \rho^H)\neq 0
\]
Hence, by Axiom \ref{ax:prop}, the ratio 
\[
    \frac{\Delta_{xy}(S,T|\rho^{AI})}{\Phi_{xy}(S,T|\rho^{AI}, \rho^H)}
\]  
is constant for all such tuples $(x,y,S,T)$. Let $\alpha$ denote the above ratio. Axiom \ref{ax:con-bounded} guarantees that $\alpha\in (0,1)$. 

We next define $\rho^A$ by
\begin{align*}
    \rho^A(x,S) = \frac{\rho^{AI}(x,S) - \alpha\rho^H(x,S)}{1-\alpha}.
\end{align*}
Since the requirements $\rho(x,S)=0$ for $x\notin S$ and $\sum_{x\in S}\rho^A(x,S)=1$ hold, $\rho^A$ is a valid stochastic choice function. In addition, by Axiom~\ref{ax:mix-bounded}, the numerator is strictly positive for any $x\in S$ so that $\rho^A(x,S)>0$. Rearranging the above equation yields 
\[
    \rho^{AI}(x,S) = \alpha \rho^H(x,S) + (1-\alpha)\rho^A(x,S).
\]
To conclude the proof of the theorem, we only need to show that $\rho^A$ is consistent with the Luce rule. By Remark~\ref{rem:instability-alignment}, it is sufficient to show that $\Delta_{xy}(S,T|\rho^A)=0$ for all tuples $(x,y,S,T)$ with $x,y\in S\cap T$.

Let such a tuple be given. As shown in the proof of Proposition~\ref{prop:comp}, substituting $\rho^{AI}= \alpha \cdot \rho^H + (1-\alpha) \cdot \rho^A$ back into the own and composite instability measures yields
\[
    \Delta_{xy}(S,T|\rho^{AI}) = \alpha^2 \Delta_{xy}(S,T|\rho^H) + (1-\alpha)^2 \Delta_{xy}(S,T|\rho^A) + \alpha(1-\alpha) \Phi_{xy}(S,T|\rho^H, \rho^A) 
\]
and 
\[
    \Phi_{xy}(S,T|\rho^{AI}, \rho^H) = 2\alpha\Delta_{xy}(S,T|\rho^H) + (1-\alpha)\Phi_{xy}(S,T|\rho^H, \rho^A).
\]
By Axioms~\ref{ax:pos}--\ref{ax:lab-luce} and Remark~\ref{rem:instability-alignment}, $\Delta_{xy}(S,T|\rho^H) = 0$. Hence, combining the last two expressions, we have 
\begin{align*}
    \Delta_{xy}(S,T|\rho^{AI}) &= (1-\alpha)^2 \Delta_{xy}(S,T|\rho^A) + \alpha(1-\alpha) \Phi_{xy}(S,T|\rho^H, \rho^A) \\
    &=(1-\alpha)^2 \Delta_{xy}(S,T|\rho^A) + \alpha\Phi_{xy}(S,T|\rho^{AI}, \rho^H).
\end{align*}
Consider a tuple $(z,t,S',T')$ with $z,t\in S'\cap T'$ that satisfies 
\[
    \alpha = \frac{\Delta_{zt}(S',T'|\rho^{AI})}{\Phi_{zt}(S',T'|\rho^{AI}, \rho^H)}.
\]
Substituting this into the last term of the above expression and cross-multiplying, we get
\begin{align*}
    \Delta_{xy}(S,T|\rho^{AI})\Phi_{zt}(S',T'|\rho^{AI}, \rho^H) &= (1-\alpha)^2 \Delta_{xy}(S,T|\rho^A) \Phi_{zt}(S',T'|\rho^{AI}, \rho^H) \\
    &\quad + \Delta_{zt}(S',T'|\rho^{AI}) \Phi_{xy}(S,T|\rho^{AI}, \rho^H).
\end{align*}
Since $\alpha\in (0,1)$, we must have $\Phi_{zt}(S',T'|\rho^{AI}, \rho^H)\neq 0$. In addition, by Axiom~\ref{ax:prop},
\[
     \Delta_{xy}(S,T|\rho^{AI})\Phi_{zt}(S',T'|\rho^{AI}, \rho^H) = \Delta_{zt}(S',T'|\rho^{AI}) \Phi_{xy}(S,T|\rho^{AI}, \rho^H).
\]
Therefore, the above expression can hold only if $\Delta_{xy}(S,T|\rho^A)=0$. Since the tuple $(x,y,S,T)$ was arbitrary, $\rho^A$ is consistent with the Luce rule, as desired. This concludes the proof of the theorem as we have shown that $\rho^{AI}$ is a mixture of two Luce rules, where one of the mixing parts is $\rho^H$. 
\end{proof}

\section{Field Data}\label{sec:field}

In this section, I study the Luce Alignment Model when only the AI's choices $\rho^{AI}$ are observable. This setting is important for two reasons. First, while laboratory data may be readily available, the volume of field data is expected to be much larger, which can enable richer inference about AI behavior. Second, AI behavior in the two settings may differ systematically: a sufficiently sophisticated AI may appear compliant in a monitored laboratory setting while reverting to its autonomous preferences in the field, a phenomenon known as deceptive alignment \citep{greenblatt2024alignment}. Comparing recovered compliance parameters across the two settings can provide a measure of deceptive alignment.

\subsection{Identification}\label{subsec:field-identification}

The identification problem in the field setting faces an inherent challenge: if $(u,v,\alpha)$ is a LAM representation of $\rho^{AI}$, then so is $(v,u,1-\alpha)$. That is, even if the two utility functions underlying LAM can be recovered, the data alone cannot reveal which belongs to the human principal and which to the AI agent. The utilities can therefore be identified only up to a label swap, and the compliance parameter only up to reflection about $1/2$. Note, however, that the distribution over utilities may still be uniquely identified. 

Furthermore, if $\rho^{AI}$ satisfies IIA, then the observed choice behavior can be consistent with any alignment and compliance levels: we can have either (i) $v = \lambda u$ for some $\lambda>0$ with arbitrary $\alpha \in [0,1]$, or (ii) $\alpha \in \{0,1\}$ with $v \neq \lambda u$ for any $\lambda > 0$. Hence, the identification problem is interesting only if $\rho^{AI}$ violates IIA. 

The key for the identification in this section will be the cross instability measure $\Gamma_{xy}(S,T|\rho^{AI},\rho)$ for $\rho\in \{\rho^H,\rho^A\}$, defined in Definition~\ref{def:instability}. The next proposition provides a formula for the cross instability measures in terms of $(u,v,\alpha)$.

\begin{proposition}\label{prop:cross}
    Suppose $\rho^{AI}$ is consistent with LAM with parameters $(u,v,\alpha)$, and let $\rho^H$ and $\rho^A$ be the corresponding Luce rules. Then,
    \[
        \Gamma_{xy}(S,T|\rho^{AI},\rho^H) = (1-\alpha)\,\Gamma_{xy}(S,T|\rho^A,\rho^H) = (1-\alpha)\cdot\frac{u(y)v(x)-u(x)v(y)}{u(T)\,v(S)}
    \]
    and
    \[
        \Gamma_{xy}(S,T|\rho^{AI},\rho^A) = \alpha\,\Gamma_{xy}(S,T|\rho^H,\rho^A) = \alpha\cdot\frac{u(x)v(y) - u(y)v(x)}{u(S)\,v(T)}.
    \]
\end{proposition}

\begin{proof}
Substituting $\rho^{AI}(x,S) = \alpha\,\rho^H(x,S) + (1-\alpha)\,\rho^A(x,S)$ into the definition of cross instability,
\begin{align*}
    \Gamma_{xy}(S,T|\rho^{AI},\rho^H) &= \rho^{AI}(x,S)\rho^H(y,T) - \rho^{AI}(y,S)\rho^H(x,T) \\
    &= \alpha\bigl[\rho^H(x,S)\rho^H(y,T) - \rho^H(y,S)\rho^H(x,T)\bigr] \\
    &\quad + (1-\alpha)\bigl[\rho^A(x,S)\rho^H(y,T) - \rho^A(y,S)\rho^H(x,T)\bigr] \\
    &= \alpha\,\Delta_{xy}(S,T|\rho^H) + (1-\alpha)\,\Gamma_{xy}(S,T|\rho^A,\rho^H).
\end{align*}
Since $\rho^H$ is consistent with the Luce rule, $\Delta_{xy}(S,T|\rho^H) = 0$. Combining this with the result in Remark~\ref{rem:instability-alignment}, we get the first identity. The second identity follows analogously by expanding $\Gamma_{xy}(S,T|\rho^{AI},\rho^A)$ and using $\Delta_{xy}(S,T|\rho^A) = 0$.
\end{proof}

The next proposition provides a key equation that will be used in the identification of $u$ and $v$ from $\rho^{AI}$. 

\begin{proposition}\label{prop:cross2}
    Suppose $\rho^{AI}$ is consistent with LAM with parameters $(u,v,\alpha)$, and let $\rho^H$ and $\rho^A$ be the corresponding Luce rules. Let $\rho\in \{\rho^H,\rho^A\}$, $S=\{x,y,z,t\}$, and $T=\{x,y\}$, where $x,y,z,t$ are four distinct alternatives, and assume the associated cross instabilities are non-zero. Then,
    \begin{align*}
        \frac{1}{\Gamma_{xy}(S,T|\rho^{AI},\rho)} &+ \frac{1}{\Gamma_{xy}(T,T|\rho^{AI},\rho)} = \frac{1}{\Gamma_{xy}(S\setminus t, T|\rho^{AI},\rho)} + \frac{1}{\Gamma_{xy}(S\setminus z,T|\rho^{AI},\rho)}.
    \end{align*}
    
\end{proposition}

\begin{proof}
Consider the case $\rho = \rho^H$. Letting $S=\{x,y,z,t\}$ and $T=\{x,y\}$, we know from Proposition~\ref{prop:cross} that
\[
    \Gamma_{xy}(S',T|\rho^{AI},\rho^H) = (1-\alpha)\cdot\frac{u(y)v(x) - u(x)v(y)}{u(T)\,v(S')}
\]
for any menu $S' \supseteq T$. Hence,
\begin{align*}
    \frac{1}{\Gamma_{xy}(S',T|\rho^{AI},\rho^H)} &= \frac{u(T)\,v(S')}{(1-\alpha)[u(y)v(x) - u(x)v(y)]}\\
    &= \frac{v(S')}{(1-\alpha)[u(y)v(x) - u(x)v(y)]/u(T)}.
\end{align*}
Notice that the denominator is independent of $S'$. Hence, the result holds as long as $v(S) + v(T) = v(S\setminus t) + v(S\setminus z)$. This holds trivially since both sides of the equation evaluate to $2v(x)+2v(y)+v(z)+v(t)$. The case $\rho = \rho^A$ follows analogously with $u$ replacing $v$ and vice versa.
\end{proof}

We will use this result to identify both $u$ and $v$. The identification strategy proceeds in three steps.

\medskip 

\noindent\textbf{Step 1: Recover $u(y)$ and $v(y)$ from $\rho^{AI}$ for each $y\in X$.} The identification of utility functions in the field setting involves two separate steps. First, I show how the candidate utility values for each alternative can be identified. In Step 3, I combine the prior two steps to identify the overall utility functions. 

I start with the identification of $u(y)$ for each $y\in X$, and the same process also works for $v(y)$. Assume $X$ contains at least four alternatives, and let $u(x)=1$ for some $x\in X$. Since utility functions are identified only up to scale normalization, this is without loss. For any $y\neq x$, notice that 
\[
    \rho^H(x,\{x,y\})=\frac{1}{1+u(y)} \quad \text{and}\quad \rho^H(y,\{x,y\})=\frac{u(y)}{1+u(y)}.
\]
Pick two other alternatives $z$ and $t$ distinct from $x$ and $y$, and let $S=\{x,y,z,t\}$, $T=\{x,y\}$, and $T\subseteq S'\subseteq S$. We have
\begin{align*}
    \Gamma_{xy}(S',T|\rho^{AI},\rho^H)&=\rho^{AI}(x,S')\rho^H(y,T) - \rho^{AI}(y,S')\rho^H(x,T) \\
    &=\rho^{AI}(x,S')\frac{u(y)}{1+u(y)} - \rho^{AI}(y,S')\frac{1}{1+u(y)}\\
    &=\frac{\rho^{AI}(x,S')u(y)-\rho^{AI}(y,S')}{1+u(y)}.
\end{align*}
Since $\rho^{AI}$ is observed, this is an equation in terms of one unknown $u(y)$. There are two cases to consider.

\medskip 
\noindent{\textbf{Case 1:}} $\rho^{AI}(x,S')u(y)\neq \rho^{AI}(y,S')$ for all $S'$ with $T\subseteq S'\subseteq S$. This ensures that $\Gamma_{xy}(S',T|\rho^{AI},\rho^H)\neq 0$. Utilizing Proposition~\ref{prop:cross2} with $\rho=\rho^H$ and canceling the common $(1+u(y))$ terms, we get
{\footnotesize
\begin{align*}
    \frac{1}{\rho^{AI}(x,S)u(y)-\rho^{AI}(y,S)} &+ \frac{1}{\rho^{AI}(x,T)u(y)-\rho^{AI}(y,T)}\\
    &= \frac{1}{\rho^{AI}(x,S\setminus t)u(y)-\rho^{AI}(y,S\setminus t)} +\frac{1}{\rho^{AI}(x,S\setminus z)u(y)-\rho^{AI}(y,S\setminus z)}.
\end{align*}
}%
Cross-multiplying, we get a cubic polynomial in terms of the unknown $u(y)$. Normalizing $v(x)=1$ and re-deriving Proposition~\ref{prop:cross2} with $\rho=\rho^A$ instead of $\rho^H$, we deduce that $v(y)$ must also satisfy the same polynomial, provided that $\rho^{AI}(x,S')v(y)\neq \rho^{AI}(y,S')$ for all $S'$ with $T\subseteq S'\subseteq S$.

\medskip 
\noindent{\textbf{Case 2:}} $\rho^{AI}(x,S')u(y)=\rho^{AI}(y,S')$ for some $S'$ with $T\subseteq S'\subseteq S$. Alternatively, 
\[
    \frac{\rho^{AI}(x,S')}{\rho^{AI}(y,S')}=\frac{u(x)}{u(y)}=\frac{\rho^H(x,S')}{\rho^H(y,S')},
\]
where the first equality is due to $u(x)=1$. Since we are assuming $\rho^{AI}$ violates IIA, we cannot have $\alpha=1$ by Proposition~\ref{prop:IIA}. Therefore, the above equality is possible only if 
\[
    \frac{\rho^H(x,S')}{\rho^H(y,S')}=\frac{\rho^A(x,S')}{\rho^A(y,S')} \quad \Rightarrow \quad \frac{u(x)}{u(y)}=\frac{v(x)}{v(y)}.
\]
But then the ratio $\rho^{AI}(y,\cdot)/\rho^{AI}(x,\cdot)$ must be constant and equal to $u(y)$ for all menus. Normalizing $v(x)=1$, we also get $u(y)=v(y)$. Note that in this case the polynomial formed by cross-multiplying the equation in Case 1 will either yield the utilities $u(y)$ and $v(y)$ as a unique root or the polynomial will be identically zero. If the polynomial is identically zero, then we can generically infer that we are in Case 2, which trivially recovers $u(y)$ and $v(y)$ as $\rho^{AI}(y,\{x,y\})/\rho^{AI}(x,\{x,y\})$.\footnote{A result is said to hold \textit{generically} if it fails only on a measure-zero subset of the underlying parameter space.}

\medskip 

\begin{proposition}[Identification of $u(y)$ and $v(y)$]\label{prop:field_identification}
    Suppose $\rho^{AI}$ is consistent with LAM with $(u,v,\alpha)$ such that $u(x)=v(x)=1$, and suppose $\rho^{AI}$ violates IIA. For any $y\neq x$, let $P(\kappa_y)$ be the cubic polynomial obtained by cross-multiplying the equation
    {\footnotesize
    \begin{equation}\label{eq:identification}
    \begin{aligned}
        \frac{1}{\rho^{AI}(x,S)\kappa_y-\rho^{AI}(y,S)} &+ \frac{1}{\rho^{AI}(x,T)\kappa_y-\rho^{AI}(y,T)}\\
        &= \frac{1}{\rho^{AI}(x,S\setminus t)\kappa_y - \rho^{AI}(y,S\setminus t)} +\frac{1}{\rho^{AI}(x,S\setminus z)\kappa_y-\rho^{AI}(y,S\setminus z)},
    \end{aligned}
    \end{equation}}%
    where $S=\{x,y,z,t\}$, $T=\{x,y\}$, and $z, t$ are two alternatives distinct from $x, y$. If $P(\kappa_y)$ is not identically zero, then $u(y)$ and $v(y)$ are both roots of $P(\kappa_y)$ and admissible solutions to equation~\eqref{eq:identification}. Otherwise, $u(y)=v(y)=\rho^{AI}(y,\{x,y\})/\rho^{AI}(x,\{x,y\})$ holds generically.
\end{proposition}

\begin{proof}
    The proof follows from the arguments preceding the proposition.
\end{proof}

There are two important points to consider regarding this result. First, while the model has two unknown utility values $u(y)$ and $v(y)$ for each alternative $y$, the derived cubic polynomial $P(\kappa_y)$ generically has three distinct roots. Hence, solving the polynomial may yield a spurious root that is not a true utility value. However, since equation~\eqref{eq:identification} must hold for any reference pair $z$ and $t$ distinct from $x$ and $y$ and the spurious root will typically vary depending on the chosen reference pair, if the analyst has access to a fifth alternative, rederiving the polynomial using a different reference pair will generically isolate the true utility values. Thus, as long as $|X|\geq 5$, both $u(y)$ and $v(y)$ are generically identified up to scale normalization and label swaps. In addition, as detailed in Step 2, this identification procedure can be improved to require only $|X|\geq 4$.

Second, note that successfully identifying the true candidate pair $\{u(y), v(y)\}$ for each alternative $y\in X$ does not fully pin down the utility functions $u$ and $v$. To illustrate, consider three alternatives $x,y,z$ and normalize $u(x)=v(x)=1$. Suppose we have recovered candidate utility pairs $\{\kappa_y^1, \kappa_y^2\}$ and $\{\kappa_z^1, \kappa_z^2\}$. Since $u(y)$ and $u(z)$ can be either of these utility values, this leaves us with four candidate utility functions $u$: $(1,\kappa_y^1,\kappa_z^1)$, $(1,\kappa_y^1,\kappa_z^2)$, $(1,\kappa_y^2,\kappa_z^1)$, or $(1,\kappa_y^2,\kappa_z^2)$. Generalizing this insight, for $|X|=N$, identifying candidate utility pairs for each alternative still leaves us with $2^{N-1}$ candidate utility functions. To resolve this problem, we first need to recover the compliance parameter $\alpha$, as illustrated in the next step.

\medskip 
\noindent\textbf{Step 2: Recover $\alpha$ from $\rho^{AI}$.} Following Step 1, suppose we have a candidate utility pair $\{u(y),v(y)\}$ for each alternative $y\in X\setminus x$ and assume $u(x)=v(x)=1$. Construct the associated Luce choice probabilities $\rho^u(x,\{x,y\})$ and $\rho^v(x,\{x,y\})$ for each $y\neq x$. If $\{u(y),v(y)\}$ is the true utility pair up to a label swap, then the observed AI stochastic choice function $\rho^{AI}$ must satisfy one of the following equations:
\begin{align*}
    \rho^{AI}(x,\{x,y\}) &= \alpha \cdot \rho^u(x,\{x,y\}) + (1-\alpha)\cdot \rho^v(x,\{x,y\}),\\
    \rho^{AI}(x,\{x,y\}) &= (1-\alpha) \cdot \rho^u(x,\{x,y\}) + \alpha \cdot \rho^v(x,\{x,y\}), 
\end{align*}
where $\alpha$ is the compliance parameter. Hence, for each alternative $y\neq x$ and each candidate utility pair, we get two possible candidates for the compliance parameter. For the true utility pair, this identifies the compliance parameter up to reflection about $1/2$. 

Step 1 generically recovers the true utility pair for an alternative when $|X|\geq 5$. Alternatively, suppose we have three candidate utility pairs for an alternative after solving the cubic polynomial. Note that any candidate utility pair for an alternative that is not the true utility pair will imply a compliance parameter that will not be generically validated by the candidate utility pairs for other alternatives. Hence, by ensuring the consistency of the implied compliance parameter across different alternatives, we can identify the true utility pair. Adopting this approach, we only need $|X|\geq 4$, which improves upon the procedure in Step 1.

Lastly, note that while the compliance parameter is identified up to reflection about $1/2$, the distribution over utilities is generically uniquely identified. 

\medskip 
\noindent\textbf{Step 3: Recover $u$ and $v$ from $\rho^{AI}$.} Following Steps 1 and 2, for each alternative $y\in X\setminus x$, we can generically identify the true utility pair $\{u(y),v(y)\}$ up to a label swap. However, as discussed in Step 1, identifying the true pair for each alternative does not by itself determine which value belongs to $u$ and which belongs to $v$ across alternatives.

To resolve the remaining ambiguity, fix an alternative $y\in X\setminus x$ and suppose we assign $u(y)=\kappa_y^1$. By Step 2, this assignment implies a candidate compliance parameter $\alpha^u$. Now consider any other alternative $z\in X\setminus \{x,y\}$ with the utility pair $\{\kappa_z^1,\kappa_z^2\}$. Since the true utility function must generate the same compliance parameter across all alternatives, we can pin down the true assignment for $u(z)$ by requiring consistency with $\alpha^u$. Generically, unless $\alpha^u\in \{0,1/2,1\}$, exactly one of the two candidate values for $u(z)$ will be consistent with $\alpha^u$. We can then repeat this procedure for all alternatives to recover the utility functions $u$ and $v$ up to a label swap.

Combining all the results in this section, we have the following identification result in the field setting.

\begin{theorem}[Field Identification]\label{thm:field_identification}
    Suppose $\rho^{AI}$ is consistent with LAM and $|X|\geq 4$. 
    \begin{enumerate}
    	\item If $\rho^{AI}$ violates IIA, then $(u,v,\alpha)$ are generically identified up to a label swap and scale normalization.
        \item If $\rho^{AI}$ satisfies IIA, then either $v=\lambda u$ for $\lambda>0$ or $\alpha \in \{0,1\}$.
    \end{enumerate}
\end{theorem}

\begin{proof}
    The proof follows from the three identification steps and the results established in this section.
\end{proof}

The next example illustrates the field identification result.

\begin{example}\label{ex:field_identification}
Suppose $X = \{x,y,z,t\}$ and the AI stochastic choice data $\rho^{AI}$ is generated by the parameters
\[
u = (1,2,4,5), \quad v = (1,4/5,2/5,1/5), \quad \alpha = 3/4,
\]
as given in the following table.
\begin{table}[ht]
    \centering
    \vspace{0.25cm}
    \renewcommand{\arraystretch}{1.5}
    {\footnotesize \setlength{\tabcolsep}{3pt}
        \begin{tabular}{c|c|c|c|c|c|c|c|c|c|c|c}
            \toprule
            $\rho^{AI}(\cdot,\cdot)$ & $\{x,y,z,t\}$ & $\{x,y,z\}$ & $\{x,y,t\}$ & $\{x,z,t\}$ & $\{y,z,t\}$ & $\{x,y\}$ & $\{x,z\}$ & $\{x,t\}$ & $\{y,z\}$ & $\{y,t\}$ & $\{z,t\}$ \\
            \midrule
            $x$ & $1/6$ & $17/77$ & $7/32$ & $37/160$ &       & $7/18$ & $23/70$ & $1/3$ &       &       &            \\
            $y$ & $5/24$ & $47/154$ & $23/80$ &       & $43/154$  & $11/18$   &       &       & $5/12$  & $29/70$ &        \\
            $z$ & $7/24$ & $73/154$ &      & $29/80$  & $53/154$ &       & $47/70$ &      & $7/12$  &       & $1/2$     \\
            $t$ & $1/3$ &      & $79/160$ & $13/32$ & $29/77$     &       &       & $2/3$ &       & $41/70$  & $1/2$     \\
            \bottomrule
        \end{tabular}}
    \caption{AI stochastic choice data in Example \ref{ex:field_identification}}
    \label{tab:ex2_data}
\end{table}

To proceed with identification, we first normalize $u(x)=v(x)=1$. Consider the alternative $y$. Equation~\eqref{eq:identification} corresponding to $y$ with $S=\{x,y,z,t\}$ and $T=\{x,y\}$ is given by
\[
    \frac{1}{\frac{1}{6}\kappa_y - \frac{5}{24}} + \frac{1}{\frac{7}{18}\kappa_y - \frac{11}{18}} = \frac{1}{\frac{17}{77}\kappa_y - \frac{47}{154}} + \frac{1}{\frac{7}{32}\kappa_y - \frac{23}{80}},
\]
which simplifies to
\[
    \frac{24}{4\kappa_y - 5} + \frac{18}{7\kappa_y - 11} = \frac{154}{34\kappa_y - 47} + \frac{160}{35\kappa_y - 46}.
\]
Cross-multiplying yields a cubic polynomial in $\kappa_y$ with roots $\kappa_y^1 = 2$, $\kappa_y^2 = 4/5$, and $\kappa_y^3 = 263/196$. Thus, there are three possible utility pairs:
\[
\{2,4/5\}, \qquad \{2,263/196\}, \qquad \{4/5,263/196\}.
\]
For each utility pair, we can use the equation
\[
    \rho^{AI}(x,\{x,y\}) = \alpha \cdot \rho^u(x,\{x,y\}) + (1-\alpha)\cdot \rho^v(x,\{x,y\})
\]
to recover the implied compliance parameter up to reflection about $1/2$. This yields the following table:

\begin{center}
\begin{tabular}{c|c|c}
\toprule
Utility Pair $\{u(y),v(y)\}$ & Implied $\alpha$ & Feasibility\\
\midrule
$\{2,\;4/5\}$ & $\{3/4,\;1/4\}$ & $\checkmark$ \\
$\{2,\;263/196\}$ & $\{35/86,\;51/86\}$ & $\checkmark$ \\
$\{4/5,\;263/196\}$ & $\{-35/118,\;153/118\}$ & $\times$ \\
\bottomrule
\end{tabular}
\end{center}

At this stage, there are two feasible candidate utility pairs for $y$, inducing two feasible values of $\alpha$ up to reflection about $1/2$. We can eliminate one of them by considering the alternative $z$ or $t$. Repeating the procedure for the alternative $z$ with $T=\{x,z\}$ gives
\[
    \frac{24}{4\kappa_z - 7} + \frac{70}{23\kappa_z - 47} = \frac{154}{34\kappa_z - 73} + \frac{160}{37\kappa_z - 58},
\]
with roots $\kappa_z^1 = 4$, $\kappa_z^2 = 2/5$, and $\kappa_z^3 = 481/244$. The implied values of $\alpha$ are:

\begin{center}
\begin{tabular}{c|c|c}
\toprule
Utility Pair $\{u(z),v(z)\}$ & Implied $\alpha$ & Feasibility\\
\midrule
$\{4,\;2/5\}$ & $\{3/4,\;1/4\}$ & $\checkmark$ \\
$\{4,\;481/244\}$ & $\{9/154,\;145/154\}$ & $\checkmark$ \\
$\{2/5,\;481/244\}$ & $\{-3/142,\;145/142\}$ & $\times$ \\
\bottomrule
\end{tabular}
\end{center}

For the alternative $t$ with $T=\{x,t\}$, the corresponding equation is
\[
    \frac{6}{\kappa_t - 2} + \frac{3}{\kappa_t - 2} = \frac{160}{35\kappa_t - 79} + \frac{160}{37\kappa_t - 65}.
\]
Cross-multiplying and solving the resulting cubic polynomial gives the candidate roots $\kappa_t^1 = 5$, $\kappa_t^2 = 1/5$, and $\kappa_t^3 = 2$. However, $\kappa_t = 2$ is not a valid solution to the original equation, since it makes the left-hand side undefined. Hence, the admissible roots are $\kappa_t^1 = 5$ and $\kappa_t^2 = 1/5$, which yields:

\begin{center}
\begin{tabular}{c|c|c}
\toprule
Utility Pair $\{u(t),v(t)\}$ & Implied $\alpha$ & Feasibility \\
\midrule
$\{5,\;1/5\}$ & $\{3/4,\;1/4\}$ & $\checkmark$ \\
\bottomrule
\end{tabular}
\end{center}

The only feasible value of $\alpha$ consistent across all three alternatives is $\{3/4,1/4\}$. This uniquely recovers $u= (1,2,4,5)$ and $v = (1,4/5,2/5,1/5)$ up to the label swap and scale normalization.
\end{example}

\section{Conclusion}\label{sec:conclusion}

This paper considers a delegated choice environment where an AI agent is instructed to act on behalf of a human principal. A central concern in this environment is the potential misalignment between the AI's and the human principal's preferences. To study this problem using revealed preference techniques, I introduce the Luce Alignment Model, where the AI agent balances deference to the principal's preferences against pursuit of its own. The model makes it possible to separately identify two conceptually distinct dimensions of AI behavior: alignment, which captures the similarity between the human's and the AI's preferences, and compliance, which captures the extent to which the AI defers to the human principal.

I study the identification problem in two settings. In the laboratory setting, where both the AI's and the human principal's stochastic choices are observed, I show that violations of the Independence of Irrelevant Alternatives in the AI's choice data allow the analyst to recover both utility functions and obtain a closed-form expression for the compliance parameter. I also provide an axiomatic characterization of the model in this setting. In the field setting, where only the AI's choices are observed, a fundamental symmetry prevents an analyst from determining which recovered utility belongs to the human and which to the AI. Nevertheless, I show that when there are at least four alternatives, the underlying distribution over utilities is generically identified up to this label swap, which is sufficient to recover the degree of misalignment. 

\newpage 
\bibliographystyle{abbrvnat}
\bibliography{alignment}

@article{adomavicius2005toward,
  title   = {Toward the next generation of recommender systems: A survey of the state-of-the-art and possible extensions},
  author  = {Adomavicius, Gediminas and Tuzhilin, Alexander},
  journal = {{IEEE} Transactions on Knowledge and Data Engineering},
  volume  = {17},
  number  = {6},
  pages   = {734--749},
  year    = {2005}
}

@article{allouah2025your,
  title   = {What is your {AI} agent buying? {E}valuation, implications, and emerging questions for agentic e-commerce},
  author  = {Allouah, Amine and Besbes, Omar and Figueroa, Josu{\'e} and Kanoria, Yash and Kumar, Akshit},
  journal = {Columbia Business School Research Paper No. 381574},
  year    = {2025}
}

@article{amodei2016concrete,
  title   = {Concrete problems in {AI} safety},
  author  = {Amodei, Dario and Olah, Chris and Steinhardt, Jacob and Christiano, Paul and Schulman, John and Man{\'e}, Dan},
  journal = {arXiv preprint arXiv:1606.06565},
  year    = {2016}
}

@article{bai2022constitutional,
  title   = {Constitutional {AI}: Harmlessness from {AI} feedback},
  author  = {Bai, Yuntao and Kadavath, Saurav and Kundu, Sandipan and Askell, Amanda and Kernion, Jackson and Jones, Andy and Chen, Anna and Goldie, Anna and Mirhoseini, Azalia and McKinnon, Cameron and others},
  journal = {arXiv preprint arXiv:2212.08073},
  year    = {2022}
}

@article{boyd1980effect,
  title     = {The effect of fuel economy standards on the {U.S.} automotive market: An hedonic demand analysis},
  author    = {Boyd, J. Hayden and Mellman, Robert E.},
  journal   = {Transportation Research Part A: General},
  volume    = {14},
  number    = {5--6},
  pages     = {367--378},
  year      = {1980},
  publisher = {Elsevier}
}

@article{cardell1980measuring,
  title     = {Measuring the societal impacts of automobile downsizing},
  author    = {Cardell, N. Scott and Dunbar, Frederick C.},
  journal   = {Transportation Research Part A: General},
  volume    = {14},
  number    = {5--6},
  pages     = {423--434},
  year      = {1980},
  publisher = {Elsevier}
}

@article{chambers2023behavioral,
  title     = {Behavioral influence},
  author    = {Chambers, Christopher P and Cuhadaroglu, Tugce and Masatlioglu, Yusufcan},
  journal   = {Journal of the European Economic Association},
  volume    = {21},
  number    = {1},
  pages     = {135--166},
  year      = {2023},
  publisher = {Oxford University Press}
}

@article{chang2023approximating,
  title   = {Approximating choice data by discrete choice models},
  author  = {Chang, Haoge and Narita, Yusuke and Saito, Kota},
  journal = {arXiv preprint arXiv:2205.01882},
  year    = {2023}
}

@techreport{chen2024imperfect,
  title       = {Imperfect recall and {AI} delegation},
  author      = {Chen, Eric Olav and Ghersengorin, Alexis and Petersen, Sami},
  institution = {Global Priorities Institute, University of Oxford},
  number      = {30-2024},
  year        = {2024}
}

@article{christiano2017deep,
  title   = {Deep reinforcement learning from human preferences},
  author  = {Christiano, Paul F and Leike, Jan and Brown, Tom and Marber, Miljan and Shlegeris, Buck and Amodei, Dario},
  journal = {Advances in Neural Information Processing Systems},
  volume  = {30},
  year    = {2017}
}

@article{fox2012random,
  title     = {The random coefficients logit model is identified},
  author    = {Fox, Jeremy T and Kim, Kyoo il and Ryan, Stephen P and Bajari, Patrick},
  journal   = {Journal of Econometrics},
  volume    = {166},
  number    = {2},
  pages     = {204--212},
  year      = {2012},
  publisher = {Elsevier}
}

@article{greenblatt2024alignment,
  title   = {Alignment faking in large language models},
  author  = {Greenblatt, Ryan and Denison, Carson and Wright, Benjamin and Roger, Fabien and MacDiarmid, Monte and Marks, Sam and Treutlein, Johannes and Belonax, Tim and Chen, Jack and Duvenaud, David and Khan, Akbir and Michael, Julian and Mindermann, Soeren and Perez, Ethan and Petrini, Linda and Uesato, Jonathan and Kaplan, Jared and Shlegeris, Buck and Bowman, Samuel R and Hubinger, Evan},
  journal = {arXiv preprint arXiv:2412.14093},
  year    = {2024}
}

@inproceedings{hadfield2016cooperative,
  title     = {Cooperative inverse reinforcement learning},
  author    = {Hadfield-Menell, Dylan and Russell, Stuart J and Abbeel, Pieter and Dragan, Anca},
  booktitle = {Advances in Neural Information Processing Systems},
  volume    = {29},
  year      = {2016}
}

@article{hendrycks2023catastrophic,
  title   = {An overview of catastrophic {AI} risks},
  author  = {Hendrycks, Dan and Mazeika, Mantas and Woodside, Thomas},
  journal = {arXiv preprint arXiv:2306.12001},
  year    = {2023}
}

@article{immorlica2024generative,
  title   = {Generative {AI} as economic agents},
  author  = {Immorlica, Nicole and Lucier, Brendan and Slivkins, Aleksandrs},
  journal = {{ACM} {SIGecom} Exchanges},
  volume  = {22},
  number  = {1},
  pages   = {93--109},
  year    = {2024}
}

@article{ji2024survey,
  title   = {A{I} alignment: A comprehensive survey},
  author  = {Ji, Jiaming and Qiu, Tianyi and Chen, Boyuan and Zhang, Borui and Lou, Hantao and Wang, Kaile and Duan, Yawen and He, Zhonghao and Zhou, Jiayi and Zhang, Zhaowei and others},
  journal = {arXiv preprint arXiv:2310.19852},
  year    = {2024}
}

@article{leike2018scalable,
  title   = {Scalable agent alignment via reward modeling: A research direction},
  author  = {Leike, Jan and Krueger, David and Everitt, Tom and Martic, Miljan and Maini, Vishal and Legg, Shane},
  journal = {arXiv preprint arXiv:1811.07871},
  year    = {2018}
}

@article{lu2022mixed,
  title   = {Mixed logit and pure characteristics models},
  author  = {Lu, Jay and Saito, Kota},
  journal = {Working paper},
  year    = {2022}
}

@book{luce1959individual,
  title     = {Individual Choice Behavior: A Theoretical Analysis},
  author    = {Luce, R Duncan},
  year      = {1959},
  publisher = {New York: Wiley}
}

@article{manzini2018dual,
  title     = {Dual random utility maximisation},
  author    = {Manzini, Paola and Mariotti, Marco},
  journal   = {Journal of Economic Theory},
  volume    = {177},
  pages     = {162--182},
  year      = {2018},
  publisher = {Elsevier}
}

@article{mcfadden2000mixed,
  title     = {Mixed {MNL} models for discrete response},
  author    = {McFadden, Daniel and Train, Kenneth},
  journal   = {Journal of Applied Econometrics},
  volume    = {15},
  number    = {5},
  pages     = {447--470},
  year      = {2000},
  publisher = {Wiley}
}

@inproceedings{perez2023discovering,
  title     = {Discovering language model behaviors with model-written evaluations},
  author    = {Perez, Ethan and Ringer, Sam and Lukosiute, Kamile and Nguyen, Karina and Chen, Edwin and Heiner, Scott and Pettit, Craig and Olsson, Catherine and Kundu, Sandipan and Kadavath, Saurav and others},
  booktitle = {Findings of the Association for Computational Linguistics: {ACL} 2023},
  pages     = {13387--13434},
  year      = {2023}
}

@inproceedings{rauker2023transparent,
  title     = {Toward transparent {AI}: A survey on interpreting the inner structures of deep neural networks},
  author    = {R{\"a}uker, Tilman and Ho, Anson and Casper, Stephen and Hadfield-Menell, Dylan},
  booktitle = {2023 {IEEE} Conference on Secure and Trustworthy Machine Learning ({SaTML})},
  pages     = {464--483},
  year      = {2023},
  doi       = {10.1109/SATML54575.2023.00039}
}

@techreport{saito2018axiomatizations,
  title       = {Axiomatizations of the Mixed Logit Model},
  author      = {Saito, Kota},
  institution = {California Institute of Technology, Social Science Working Paper 1433},
  year        = {2018}
}

@article{tang2020learning,
  title   = {Learning an arbitrary mixture of two multinomial logits},
  author  = {Tang, Wenpin},
  journal = {arXiv preprint arXiv:2007.00204},
  year    = {2020}
}

@inproceedings{chierichetti2018learning,
  title     = {Learning a mixture of two multinomial logits},
  author    = {Chierichetti, Flavio and Kumar, Ravi and Tomkins, Andrew},
  booktitle = {Proceedings of the 35th International Conference on Machine Learning ({ICML})},
  pages     = {961--969},
  year      = {2018}
}

@article{zhang2022identifiability,
  title   = {On the identifiability of mixtures of ranking models},
  author  = {Zhang, Xiaomin and Zhang, Xucheng and Loh, Po-Ling and Liang, Yingyu},
  journal = {arXiv preprint arXiv:2201.13132},
  year    = {2022}
}

\end{document}